\documentclass[leqno,a4paper,12pt]{article}
\usepackage[centertags]{amsmath}
\usepackage{amsfonts}
\usepackage{amssymb}
\usepackage{multirow}
\usepackage{amsthm}
\usepackage{nicefrac}
\usepackage{pifont}
\usepackage[ansinew]{inputenc}
\usepackage{graphicx}
\usepackage[arrow, matrix, curve]{xy}

\theoremstyle{plain}
\newtheorem{thm}{Proposition}[]
\newtheorem{rem}{Remark}[]

\begin{document}

\title{On Keiding's Equation and its relation to differential equations about prevalence and incidence in chronic disease epidemiology}
\author{Ralph Brinks}
\date{}
\maketitle

\begin{abstract}
We study the relation between the age-specific prevalence, incidence and mortality
in an illness-death model consisting of the three states \emph{Healthy, Ill, Dead}. The dependency
on three different time scales (age, calendar time, disease duration) is considered. 
It is shown that Keiding's equation published in 1991
is a generalisation of the solution of Brunet and Struchiner's partial differential 
equation from 1999. In a special case, we propose a particularly simple estimate of the incidence 
from prevalence data.
\end{abstract}

\section{Background}
Keiding reviewed the relations between the incidence and prevalence of a chronic disease based on
an illness-death model \cite{Kei91}. The illness-death model consists of the three states 
\emph{Healthy, Ill and Dead} (Figure \ref{fig:CompModel}). The transition rates $i, m_0,$ and $m_1$
between the states may depend on the time scales calendar time ($t$), age ($a$), and the rate $m_1$
may additionally depend on the duration of the disease ($d$). As we are dealing with chronic 
diseases, there is no transition from the state \emph{Ill} to \emph{Healthy}.
Let $S(t, a)$ denote the number of persons aged $a, ~a \ge 0$ at time $t$ in the state \emph{Healthy}.
Similarly, $C(t, a, d)$ denotes the number of persons aged $a$ at $t$ who are diseased for $d, ~d\ge 0$
time units. The notation is chosen for historical reasons, $S$ and $C$ stand for susceptibles and cases, 
respectively. 

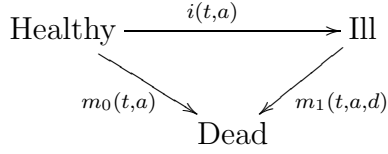
\begin{figure}[ht]
\centering
$$
\begin{xy}
\xymatrix{
\text{Healthy}\ar[rr]^{i(t,a)}\ar[rd]_{m_0(t,a)} &   & \text{Ill}\ar[dl]^{m_1(t,a,d)} \\
                                    & \text{Dead} &
}
\end{xy}
$$
\caption{Illness-death model. The transition rates $i$ and $m_0$ depend on calendar 
time $t$ and age $a.$ The rate $m_1$  additionally depends on the duration $d.$}
\label{fig:CompModel}
\end{figure}

In epidemiology, it is common to consider the age-specific prevalence
$$p(t, a) = \frac{C^\star(t, a)}{C^\star(t, a) + S(t, a)},$$
where
$C^\star (t, a) = \int_0^a C(t, a, \delta) \, \mathrm{d}\delta$ denotes the number of
diseased persons aged $a$ at time $t,$ irrespective of the duration $d.$ Keiding gave following
expression for the prevalence odds \cite[p. 379]{Kei91}:
\begin{equation}\label{e:KeidingOdds}
\frac{p(t, a)}{1 - p(t, a)} = \frac{\int_0^a \mathcal{M}_{t,a}(y) \, i(t-a+y, y) \, 
                                             e^{- \int_y^a m_1(t-a+\tau, \tau, \tau-y) \mathrm{d}\tau}
                                           \, \mathrm{d}y
                                   }{\mathcal{M}_{t,a}(a)}
\end{equation}

with

\begin{equation*}
\mathcal{M}_{t, a}(y) = \exp \left ( -\int_0^y m_0(t-a+\tau, \tau) + i(t-a+\tau, \tau) \mathrm{d}\tau \right ).
\end{equation*}

From Equation \eqref{e:KeidingOdds} the following Proposition can be deduced.

\begin{thm}
For the age-specific prevalence $p(t, a)$ it holds
\begin{equation}
p(t, a) = \frac{\int\limits_0^a i(t-\delta, a-\delta) \, \mathcal{M}_{t,a}(a-\delta) \, 
          e^{- M_1(t, a, \delta)} \, \mathrm{d}\delta}
          {\mathcal{M}_{t,a}(a) + \int\limits_0^a i(t-\delta, a-\delta) \, \mathcal{M}_{t,a}(a-\delta) \, 
          e^{- M_1(t, a, \delta)} \, \mathrm{d}\delta}, \label{e:p}
\end{equation}
where
$$M_1(t, a, d) := \int_{0}^{d} m_1(t - d + \tau, \,a - d + \tau , \,\tau ) \, 
\mathrm{d}\tau.$$
\end{thm}
\begin{proof}
Solving Equation \eqref{e:KeidingOdds} for $p(t, a)$ and re-parametrising the path of integration 
yields Eq. \eqref{e:p}.
\end{proof}

Keiding has not presented a proof of Equation \eqref{e:KeidingOdds}. In this article,
we will give a proof and relate Equation \ref{e:p} to two partial differential equations (PDEs)
published a few years after Keiding's pivotal work in 1991.

\section{Partial differential equations}
In this section, we will formulate PDEs for $S(t, a)$ and 
$C(t, a, d)$ based on the model in Figure \ref{fig:CompModel}. The only assumptions are
\begin{itemize}
  \item All newborns are disease-free at time of birth (i.e., $C^\star(t, 0) = 0$ for all $t.$)
  \item There is no migration into or from the states \emph{Healthy} and \emph{Ill}.
  \item The rates $i, m_0$ and $m_1$ are smooth, i.e. partially differentiable 
  with continuous derivatives.
\end{itemize}

\bigskip

For the number $S$ of susceptibles we obtain following PDE:
\begin{eqnarray}
(\partial_t + \partial_a) \, S(t, a) & = & - \left ( m_0(t, a) + i(t, a) \right ) 
\, S(t, a) \label{e:PDE_S_ta} \\
S(t - a, 0) & = & S_0(t - a). \nonumber
\end{eqnarray}
Here $S_0(t - a) = S(t-a, 0)$ denotes the number of (healthy) newborns at time $t-a$. 
The notation $\partial_x$ means the partial derivative for $x, ~x \in \{t, a\}.$
Equation \eqref{e:PDE_S_ta} together with the initial condition $S_0(t - a) = S(t-a, 0)$ is
a Cauchy problem which has a unique solution (the rates $m_0$ and $i$ are  
smooth) \cite{Pol00}. This solution of the Cauchy problem is given in Eq. \eqref{e:S}.
\begin{equation}
S(t, a) = S_0(t - a) \, \exp \left ( - \int_0^a m_0(t-a+\tau, \tau) + 
i(t-a+\tau, \tau) \, \mathrm{d}\tau \right ) \label{e:S}.
\end{equation}

\bigskip

The calculation of the number $C$ of cases will be a bit more difficult, because at any 
time $t$ and at any age $a$ the current disease duration $d$ plays an important role. 
As there is no migration, the number $ C(t, a, d) $ is described by the following equations:
\begin{eqnarray*}
   C(t, a, d) &=& C(t-d, a-d, 0) \, \exp 
                        \left ( - \int_0^d m_1(t-d+\tau, a-d+\tau, \tau) \, \mathrm{d}\tau 
                        \right )\\
              &=& i(t-d, a-d) \,  S(t-d, a-d) \, 
                        e^{- \int\limits_{0}^{d} m_1(t-d+\tau, a-d+\tau, \tau) \,
                        \mathrm{d}\tau}.
\end{eqnarray*}

$C(t, a, d)$ is a solution of another Cauchy problem. The associated PDE is 
\begin{equation*}
  (\partial_t + \partial_a + \partial_d) \, C(t,a,d) = 
   - C(t, a, d) \, m_1(t, a, d), \label{e:PDE_tad}
\end{equation*}
and the initial condition is $C(t, a, 0) = i(t,a) \, S(t, a)$ for all $t, a.$
\begin{proof}
It holds
\begin{eqnarray*}
\partial_x C(t,a,d) 
  &=& \quad \partial_x i(t - d, \,a - d) \, S(t - d, \,a - d) \, \exp \left \{-M_1(t, a, d) \right \}\\
  & & + \, i(t - d, \,a - d) \, \partial_x S(t - d, \,a - d) \, \exp \left \{-M_1(t, a, d) \right \} \\ 
  & & - \, i(t - d, \,a - d)\, S(t - d, \,a - d) \, \, \exp \left \{-M_1(t, a, d) \right \} \times \\ 
  & & \qquad \qquad \qquad \qquad \qquad \qquad \qquad \qquad \qquad \partial_x M_1(t, a, d)
\end{eqnarray*}

and 

\begin{eqnarray*}
\partial_d C(t,a,d) 
   &=& - (\partial_t + \partial_a) i(t - d, \,a - d) \, S(t - d, \,a - d) \, \exp \left \{-M_1(t, a, d) \right \} \\
   & & - \, i(t - d, \,a - d) \, (\partial_t + \partial_a) S(t - d, \,a - d) \, \exp \left \{-M_1(t, a, d) \right \} \\ 
  & & - \, i(t - d, \,a - d)\, S(t - d, \,a - d) \, \, \exp \left \{-M_1(t, a, d) \right \} \times \\ 
  & & \qquad \qquad \qquad \qquad \qquad \qquad \qquad \qquad \qquad \quad \partial_d M_1(t, a, d).
\end{eqnarray*}

This implies
\begin{equation*}
(\partial_t + \partial_a + \partial_d) \, C(t,a,d) 
  = - C(t, a, d) (\partial_t + \partial_a + \partial_d) 
      \, M_1(t, a, d).
\end{equation*}

For $x \in \{t, a \}$ it is
\begin{equation*}
   \partial_x M_1(t, a, d) = \int_{0}^{d} \partial_x 
                m_1(t - d + \tau , \,a - d + \tau , \,\tau )\,\mathrm{d}\tau.
\end{equation*}

Furthermore, we find that
\begin{equation*}
   \partial_d M_1(t, a, d) = - \int_{0}^{d} (\partial_t + \partial_a) \, 
                              m_1(t - d + \tau , \,a - d + \tau , \,\tau ) \, 
                              \mathrm{d}\tau + m_1(t,a,d).
\end{equation*}
With the smoothness constraint, this proves that $C(t, a, d)$ is the unique solution of the Cauchy problem.
\end{proof}

We are interested in the overall number $C^\star(t, a):$ 
\begin{eqnarray}
C^\star(t, a) &=& \int_0^a C(t, a, \delta) \, \mathrm{d}\delta \nonumber \\
        &=& \int_0^a i(t-\delta, a-\delta) \,  S(t-\delta, a-\delta) \, 
            e^{- \int\limits_{0}^{\delta } m_1(t-\delta+\tau, a-\delta+\tau, \tau) 
            \mathrm{d}\tau} \mathrm{d}\delta \label{e:C}
\end{eqnarray}

By inserting \eqref{e:S} and \eqref{e:C} into the definition of $p(t, a),$ we obtain
Equation \eqref{e:p}. As described above, Equation \eqref{e:p} be transformed into 
Equation \eqref{e:KeidingOdds}, which proves Keiding's Equation.

\bigskip

The advantage of the Equations \eqref{e:p} and \eqref{e:KeidingOdds} 
is that for given incidence $i(t, a)$ 
and mortality rates $m_0(t, a)$ and $m_1(t, a, d),$ the age-specific prevalence can 
be calculated for all times $t$ and ages $a \ge 0.$ By this, we may estimate the impact
of health related interventions with an appropriate treatment of the involved time scales.
Unfortunately, the theory suggested by Keiding has rarely been used in epidemiology, public
health, or health economics. For instance, instead of treating time as a continuous
variable, discrete time steps are preferred, which may impose a considerable 
discretisation error (for an example of a discretisation error of more than 100\%,
see \cite{Bri14}). In the article \cite{Bri15}, the effect
of a health related intervention is estimated by treating time continuously. 

\bigskip

As a byproduct from Equation \eqref{e:p} we may conclude:
\begin{rem}
The prevalence $p(t, a)$ is independent from the number of newborns $S_0$.
\end{rem}

\bigskip

\section{Independence from the duration of the disease}
In case the mortality $m_1$ of the diseased persons is independent from the duration $d,$
the number of cases $C^\star(t, a)$ is a solution of the following PDE:
\begin{equation}
   (\partial_t + \partial_a) \, \gamma(t, a) = - m_1(t, a) \, \gamma(t, a) + i(t, a) 
                                          \, S(t, a). \label{e:PDE_C_ta}
\end{equation}
\begin{proof}
Together with the initial condition $\gamma(t-a, 0) = 0$
the PDE \eqref{e:PDE_C_ta} has the solution
\begin{eqnarray*}
\gamma(t, a) &=& e^{-\int\limits_0^a m_1(t-a+\alpha, \alpha) \, \mathrm{d}\alpha } \, \biggl \{ \gamma(t-a, 0) + \\
        &  &    ~~~~\int\limits_0^a i(t-a+\alpha, \alpha) \, S(t-a+\alpha, \alpha) \, 
              e^{\, \int\limits_0^\alpha m_1(t-a+\tau, \tau) \, \mathrm{d}\tau} 
              \mathrm{d}\alpha \biggr \}\\
        &=& \int\limits_0^a i(t-a+\alpha, \alpha) \, S(t-a+\alpha, \alpha) \, 
              e^{ -\int\limits_\alpha^a m_1(t-a+\tau, \tau) \, \mathrm{d} \tau} \mathrm{d}\alpha \\
        &=& \int\limits_0^a i(t-\delta, a-\delta) \, S(t-\delta, a-\delta) \, 
              e^{-\int\limits_{a-\delta}^a m_1(t-a+\tau, \tau) \, \mathrm{d}\tau} \mathrm{d}\delta\\
        &=& \int\limits_0^a i(t-\delta, a-\delta) \, S(t-\delta, a-\delta) \, 
              e^{-\int\limits_0^{\delta} m_1(t-\delta+\tau, a-\delta+\tau) \, \mathrm{d}\tau} \mathrm{d}\delta.
\end{eqnarray*}
By comparison with Eq. \eqref{e:C} we see that $C^\star(t, a)$ is the solution of the PDE. 
\end{proof}

If we insert \eqref{e:PDE_S_ta} and \eqref{e:PDE_C_ta} into the definition of $p(t,a),$
we may deduce following PDE \cite{Bri14}:
\begin{equation}\label{e:PDEBrinks}
(\partial_t + \partial_a) \, p = (1-p) \, \bigl ( i - p\, (m_1- m_0) \bigr).
\end{equation}

Similarly, we obtain following PDE 
for the prevalence odds $\pi(t, a) = \tfrac{p(t, a)}{1- p(t, a)}$ 
of Brunet and Struchiner \cite{Bru99}, which is equivalent to Eq. \eqref{e:PDEBrinks}:
\begin{equation}\label{e:PDEBrunet}
(\partial_t + \partial_a) \, \pi = i - \pi \, (m_1 - m_0 - i).
\end{equation}

In contrast to the PDE \eqref{e:PDEBrinks}, the PDE \eqref{e:PDEBrunet} has the
advantage of being linear. Thus, its solution is straightforward and
allows a handy simplification of Eq. \eqref{e:p} (see \cite[Eq. (1)]{Bri15a}).

\bigskip

We conclude this section with the observation that Keiding's Equation 
\eqref{e:KeidingOdds} is a generalisation of the solution of both PDEs 
\eqref{e:PDEBrinks} and \eqref{e:PDEBrunet}.

\section{Incidence being independent from calendar time}
An important application of the theory in epidemiology is the question
if incidence rate can be recovered from observed prevalence data. This question has 
already been mentioned in 1934 \cite{Mue34} and has been studied in \cite{Bri15b} 
with test data. Now it is shown that in the special case of incidence being
independent from calendar time $i(t, a) = i(a)$ the dependence of $m_1$
on the duration $d$ does not have to be known to estimate the incidence. 
This has the advantage that a possible duration dependency in $m_1$ may be
unknown.

\bigskip

Starting from \eqref{e:S} we find
\begin{eqnarray*}
	I(t, a) &:=& \int\limits_0^a i(t-a+\tau, \tau) \mathrm{d}\tau = \ln \frac{S_0(t-a)}{S(t, a)} - M_0(t, a)\\
	        &=& \ln S_0(t-a) - \ln S(t,a) - M_0(t,a)\\
	        &=& \ln S_0(t-a) - \ln \bigl (1-p(t,a) \bigr ) - \ln N(t, a) - M_0(t,a)\\
\end{eqnarray*}
with $$M_0(t, a) := \int\limits_0^a m_0(t-a+\tau, \tau) \mathrm{d}\tau.$$ The number $N(t, a)$ 
denotes the amount of persons aged $a$ at $t$ who are alive ($N = S + C.$) 
If $i$ is independent from $t$, it holds 
$$\partial_a I(t, a) = i(a) ~~\text{for all}~t.$$ 

Hence, we may deduce following representation of the age-specific incidence:
\begin{equation}\label{e:InzKZU}
	i(a) = \partial_a \biggl ( \ln S_0(t-a) - \ln \bigl (1-p(t,a) \bigr ) - \ln N(t, a) - M_0(t,a) \biggr ).
\end{equation}

This is an amazing result, because the occurring variables $S_0$ and $N$ are well known from 
demography. Assumed that the mortality $m_0$ can also be surveyed, the possibly complex 
$m_1(t, a, d)$ does not have to be known for an estimate of the incidence in case of
a given age-specific prevalence $p(t, a).$

\begin{rem}
Many epidemiological studies examine the mortality $m_1$ of the diseased instead of the
mortality $m_0$ of the non-diseased. Equation \eqref{e:InzKZU} suggests a \emph{paradoxic
study design}: Instead of following up on mortality of the diseased persons, the healthy 
persons are of primary interest.
\end{rem}

\section{Examples}

\subsection{General case}
In this subsection the age-specific prevalence $p(t, a)$ for a
hypothetical chronic disease is calculated using Equation \eqref{e:p}.
We assume that $t, a$ and $d$ are counted in units ``years'' with $t, a, d \ge 0.$
Mortality $m_0$ is assumed to be of Gompertz-Makeham type,
$$m_0(t, a) = \exp(-10.7 + 0.1 \, a) \, (1 - 0.002)^t,$$ 
and the incidence is given by
\begin{equation}\label{e:incExample}
i(t, a) = \frac{(a – 30)_+}{3000} \, (1 - 0.003)^t.
\end{equation}
The mortality $m_1$ of the diseased is assumed to be the product of $m_0$ and 
relative mortality $R(d) = (0.2 \, d - 1)^2 + 1$: 
$$m_1(t, a, d) =  R(d) \, m_0(t, a).$$ 

\bigskip

The integrals $\mathcal{M}$ and $M_1$ are calculated analytically, which
is possible here. The integral from $0$ to $a$ in the numerator and
denominator in \eqref{e:p} are calculated by Romberg's rule, which allows
an a-priori prescribed accuracy.

\bigskip

Figure \ref{fig:Results1} shows the resulting age-specific prevalences 
at $t=0, 50,$ and $100$ (in years). The three age profiles have a similar
qualitative behaviour: After onset of the disease for $a \ge 30$, the
prevalence increases sharply with age and until the seventh decade
of life. All three curves reach their maximum at the age of about 80 
(years) and then decrease slightly. 

\begin{figure}[ht]
  \centering
  \includegraphics[keepaspectratio,width=0.85\textwidth]{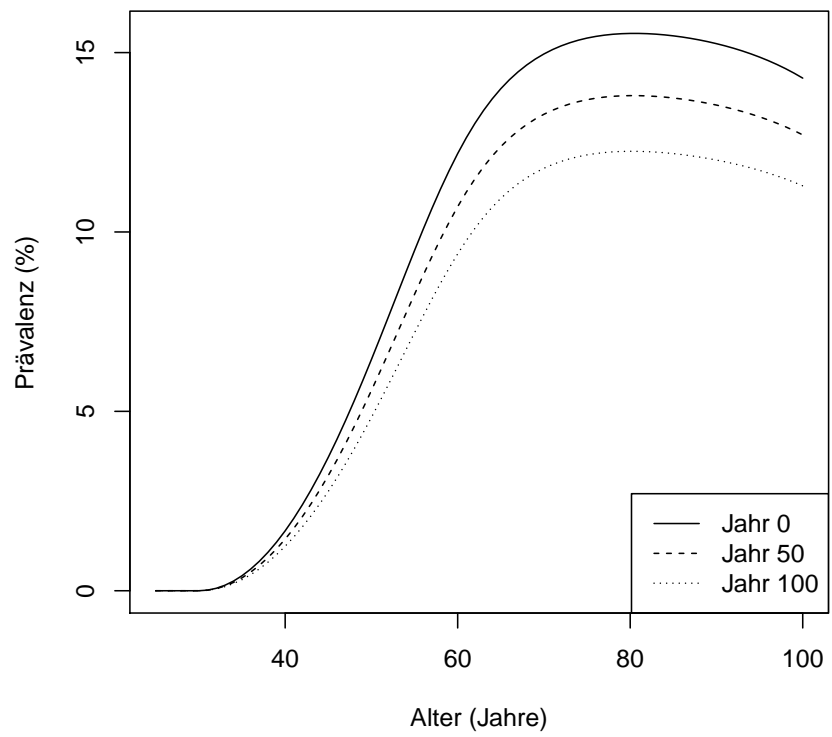}
\caption{Age-specific prevalences in the example in the years $t=0, 50,$ and $100$.} 
\label{fig:Results1}
\end{figure}

\clearpage

\subsection{Time-independent incidence}
If we leave out the term $(1-0.003)^t$ in Eq. \eqref{e:incExample}, 
we can estimate $i = i(a)$ from $p$ surveyed in year $t=100$ via Eq. \eqref{e:InzKZU}.
The partial derivative $\partial_a$ has been approximated by
a finite difference. Figure \ref{fig:Results2} show the results. 
Visually, there is a nearly perfect agreement between the theoretical and the estimated
incidence.

\begin{figure}[ht]
  \centering
  \includegraphics[keepaspectratio,width=0.85\textwidth]{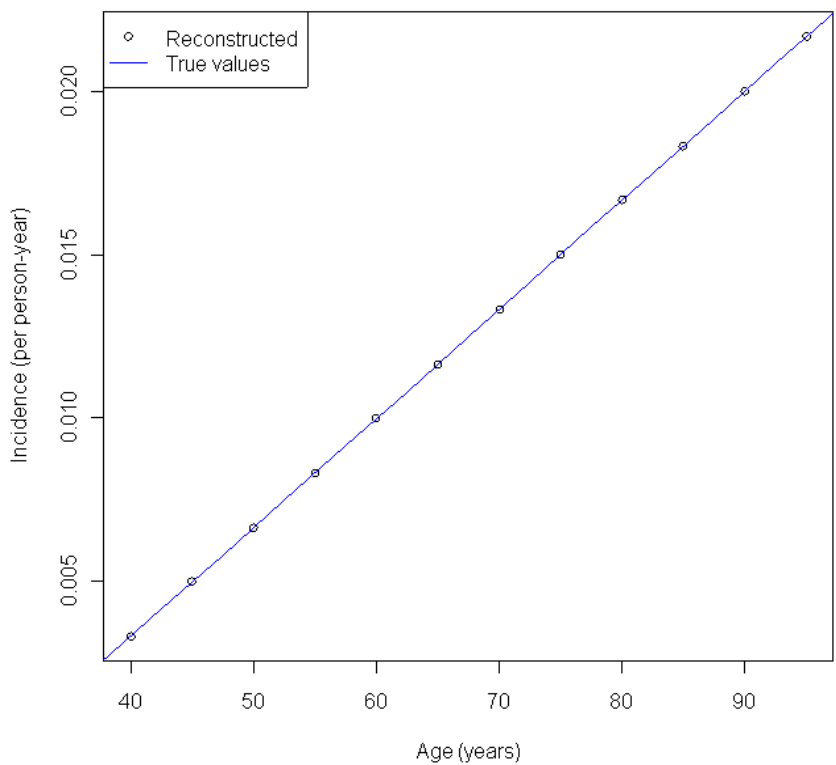}
\caption{Age-specific incidence in year $t = 100$. The solid line shows the theoretical
incidence rate $i(a) = \tfrac{(a – 30)_+}{3000}.$ The points represent the estimated values
using Equation \eqref{e:InzKZU}.} 
\label{fig:Results2}
\end{figure}

Additionally, we set up a population with a birth rate of 5000 persons per year in 
60 consecutive years (0, \dots, 59). Events in the illness-death model (diagnosis, 
death with or without the disease) are simulated by a discrete event simulation as 
described in \cite{Bri14a}. In the year $t=100$, we mimic a cross-section to
estimate the prevalence $p(100, a).$ As above, the incidence is estimated by
Eq. \eqref{e:InzKZU} and approximating $\partial_a$ by a finite difference. 
Figure \ref{fig:Results3} shows the results. In contrast to Figure \ref{fig:Results2}, 
the incidence cannot be estimated exactly with is due to the random error in the 
prevalence $p(100, a).$

\begin{figure}[ht]
  \centering
  \includegraphics[keepaspectratio,width=0.85\textwidth]{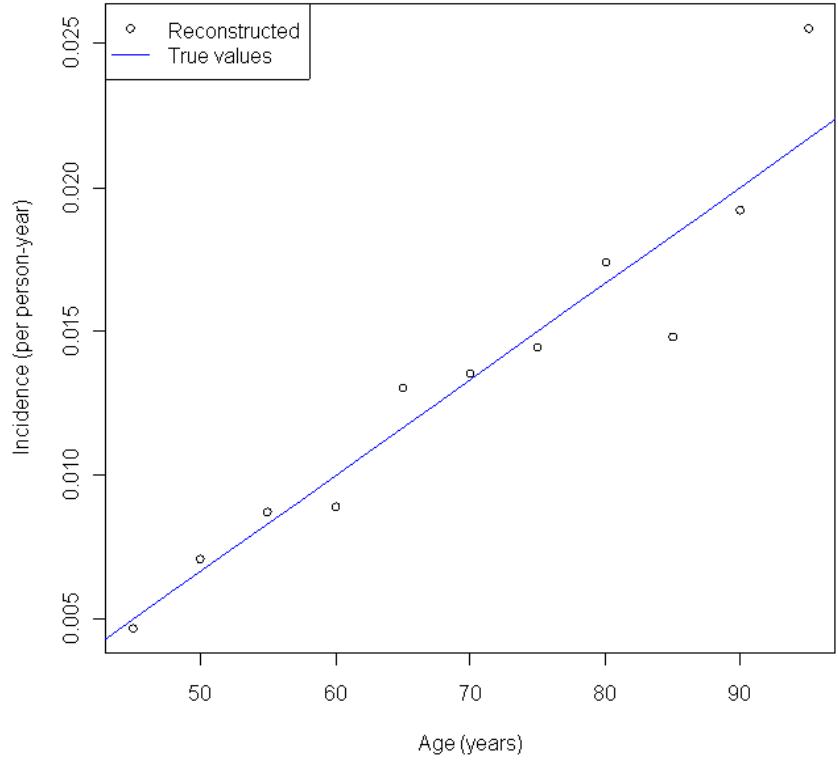}
\caption{Age-specific incidence in year $t = 100$. The solid line shows the theoretical
incidence rate $i(a) = \tfrac{(a – 30)_+}{3000}.$ The points represent the estimated values
using Equation \eqref{e:InzKZU} and the simulated prevalence.} 
\label{fig:Results3}
\end{figure}

\clearpage

\section{Conclusion}
This article combines the results of Keiding \cite{Kei91}, Brunet and Struchiner 
\cite{Bru99}, and Brinks and Landwehr \cite{Bri14,Bri15a}. We have found that
Keiding gave an analytical expression for the age-specific
prevalence in the most general case of the illness-death model, i.e. 
with involvement of all time scales (age, calendar time, and duration). 
Keiding presented this expression eight years before Brunet and Struchiner
published their linear partial differential equation without duration
dependency. Brinks and Landwehr extend the work by Keiding and Brunet and Struchiner 
by allowing migration and remission \cite{Bri14}. Even in the case with duration 
dependency, the age-specific prevalence can be related to the transition
rates in the illness-death model by a scalar partial differential equation. 
Details can be found in \cite{Bri15b}.

\bigskip

In addition, we have proposed a new way of estimating the incidence from 
a cross-sectional prevalence study where it is not necessary to 
survey the possibly complex duration dependency of $m_1.$ In this 
\emph{paradoxic study design}, the mortality of the healthy ($m_0$) needs to be known
instead of the mortality of the diseased ($m_1$). The proposed method was demonstrated by 
an example of a hypothetical chronic disease.


{}

\end{document}